\documentclass{ifacconf}

\usepackage{graphicx}      % include this line if your document contains figures
\usepackage{natbib}        % required for bibliography
\usepackage[utf8]{inputenc} % allow utf-8 input
\usepackage[T1]{fontenc}    % use 8-bit T1 fonts
\usepackage{url}            % simple URL typesetting
\usepackage{booktabs}       % professional-quality tables
\usepackage{amsfonts}       % blackboard math symbols
\usepackage{nicefrac}       % compact symbols for 1/2, etc.
\usepackage{microtype}      % microtypography

\usepackage{algorithm}
\usepackage{algpseudocode}
\usepackage{mathptmx} 
\usepackage{times} 
\usepackage{amsmath,bm} 
\usepackage{dsfont}
\usepackage{amssymb}   
\usepackage{savesym}
\savesymbol{theoremstyle}
\usepackage{amsthm} %\usepackage{theorem} 
\restoresymbol{IFAC}{theoremstyle}
\usepackage{bbm}
\usepackage{mathtools}
\usepackage{eucal}
\usepackage{cancel}
\graphicspath{{./figures/}}
\usepackage{xcolor}

\newtheorem{theorem}{Theorem}%[section]

\newtheorem{lemma}{Lemma}

\newtheorem{proposition}{Proposition}
\newtheorem{definition}{Definition}

\newcommand{\norm}[1]{\ensuremath{\left\| #1\right\|}}
\newcommand{\pbra}[1]{\ensuremath{\left( #1\right)}}

\newcommand{\cbra}[1]{\ensuremath{\left\{ #1\right\}}}
\newcommand{\abra}[1]{\ensuremath{\left< #1\right>}}
\newcommand{\oder}[2]{\ensuremath{\frac{d #1}{d #2}}}

\begin{document}
\begin{frontmatter}

\title{Semi-linear Poisson-mediated Flocking in a Cucker-Smale Model\thanksref{footnoteinfo}} 
% We are using semi-linear poisson equation only for the 1D case
% For n-D we use simple Laplacian
% We may use semi-linear poisson in n-D but I need to look at it

\thanks[footnoteinfo]{This material is based upon work supported by the Defense Advanced Research Projects Agency (DARPA) under Agreement No. HR00111990027.
This paper has been accepted for publication in 
the 24th International Symposium on Mathematical Theory of Networks and Systems (MTNS 2020), Cambridge, UK.}

\author[First]{Christos N. Mavridis} 
\author[First]{Amoolya Tirumalai} 
\author[First]{John S. Baras}
\author[Second]{Ion Matei} 

\address[First]{Electrical and Computer Engineering Department and the Institute for Systems Research, 
	University of Maryland,	College Park, MD 20742 USA, (emails: {\tt $\{$mavridis, ast256, baras$\}$@umd.edu})}
\address[Second]{Palo Alto Research Center (PARC), Palo Alto, CA 94304 USA \\
(email: {\tt{imatei@parc.com}})}

%%%%%%%%%%%%%%%%%%%%%%%%%%%%%%%%%%%%%%%%%%%%%%%%%%%%%%%%%%%%%%%%%%%%%%%%%%%%%%%%
\begin{abstract}
We propose a family of compactly supported 
parametric interaction functions in the general 
Cucker-Smale flocking dynamics such that the  
mean-field macroscopic system of mass and momentum balance equations with
non-local damping terms
can be converted from a system of partial integro-differential equations to an augmented system of partial differential equations in a compact set.
We treat the interaction functions as Green's functions for an operator corresponding to a semi-linear Poisson equation and compute the 
density and momentum in a translating reference frame, 
i.e. one that is taken in reference to the flock's centroid. This allows us to consider the dynamics in a fixed, flock-centered compact set without loss of generality.
We approach the computation of the non-local damping
using the standard finite difference treatment of the chosen differential operator, resulting in a tridiagonal system
which can be solved quickly.
\end{abstract}


\end{frontmatter}
%===============================================================================
%%%%%%%%%%%%%%%%%%%%%%%%%%%%%%%%%%%%%%%%%%%%%%%%%%%%%%%%%%%%%%%%%%%%%%%%%%%%%%%%
% Comments by the Reviewers
%%%%%%%%%%%%%%%%%%%%%%%%%%%%%%%%%%%%%%%%%%%%%%%%%%%%%%%%%%%%%%%%%%%%%%%%%%%%%%%%

% 1. Motivation behind this choice of interaction function. 
%     Why do we state that it is expressive and can be used for parametric learning?
%     Does it capture observed interaction functions?
% 2. Specify if Proposition 3 is valid only for this specific interaction function
% 3. Formulate Section 4.1 (Proof of Asymptotic Flocking) as Theorem + Proof
% 4. Consider changing the symbol 'f' in eq. (26)-(28) because we use 'f_{xv}' for 
%     the probability distribution.

%%%%%%%%%%%%%%%%%%%%%%%%%%%%%%%%%%%%%%%%%%%%%%%%%%%%%%%%%%%%%%%%%%%%%%%%%%%%%%%%
%%%%%%%%%%%%%%%%%%%%%%%%%%%%%%%%%%%%%%%%%%%%%%%%%%%%%%%%%%%%%%%%%%%%%%%%%%%%%%%%

%%%%%%%%%%%%%%%%%%%%%%%%%%%%%%%%%%%%%%%%%%%%%%%%%%%%%%%%%%%%%%%%%%%%%%%%%%%%%%%%

%%%%%%%%%%%%%%%%%%%%%%%%%%%%%%%%%%%%%%%%%%%%%%%%%%%%%%%%%%%%%%%%%%%%%%%%%%%%%%%%
\section{INTRODUCTION}
\label{Sec:Introduction}

Collective motion of autonomous agents is a widespread phenomenon 
appearing in numerous applications 
ranging from animal herding to complex networks and social dynamics
\citep{okubo1986dynamical, cucker2007emergent, giardina2008collective}.
% , giuggioli2013stigmergy}.

In general, there are two broad approaches when investigating the underlying dynamics for flocks or swarms: 
the microscopic, particle models described by ordinary differential equations (ODEs) or stochastic differential equations, 
and the macroscopic continuum models, described by partial differential
equations (PDEs).
Agent-based models assume behavioral rules at the individual level, such as 
velocity alignment, attraction, and repulsion
\citep{cucker2007emergent, giardina2008collective, ballerini2008interaction}
and are often used in numerical simulations
% to generate simplified flocking swarms, 
and in learning schemes where the interaction rules are inferred
\citep{Matei2019}.
As the number of interacting agents gets large,
the agent-based models become computationally expensive \citep{carrillo2010particle}. Considering pairwise interactions, the growth is $O(N^2)$, where $N$ is the number of agents. 
% In some very special cases where the agents can be 'decoupled' using an order parameter or similar construction, such as in the Kuramoto model, the growth is $O(N)$ \citep{strogatz}.
As we approach the mean-field limit, it is useful to consider the probability density of the agents. Using Vlasov-like arguments
\citep{carrillo2010particle}, we can construct an equation analogous to the Fokker-Planck-Kolmogorov equation. We can then define momentum and density and construct a system of compressible hydrodynamic PDEs
\citep{carrillo2010particle, shvydkoy2017eulerian}.

In flocking dynamics 
\citep{cucker2007emergent, carrillo2010particle}, 
the velocity alignment term is not only nonlocal but can also be nonlinear \citep{shvydkoy2017eulerian, karniadakis_flockingDynamicsfractionalPDEs_2018}.
The computation of the corresponding hydrodynamic equations with nonlocal forces
becomes quite costly due to 
the approximation of the convolution integrals or integral transforms using
the various quadrature methods. The simplest `quadrature' method is the Riemann sum, whose complexity is $O(n^2)$, where $n$ is the number of grid points, when estimating a convolution integral as a convolution sum in one dimension. On the other hand, an equivalent solution may be obtained using finite differences if the interaction kernel is associated with a differential operator. If that operator can be put into a sparse form, ideally a tridiagonal form, a solution can be obtained efficiently.

In this work, we modify the classical Cucker-Smale model of nonlocal particle interaction for velocity consensus
\citep{cucker2007emergent, ha2009simple}.
We propose a family of parametric interaction functions in $\mathbb{R}^d$, $d \in \{ 1,2,3\}$,
that are Green's functions for appropriately defined linear partial 
differential operators, which allow us to speed-up computation of the nonlocal interaction terms.
We investigate the conditions under which time-asymptotic flocking is achieved in the microscopic formulation in a centroid-fixed frame.
We solve the macroscopic formulation using the Kurganov-Tadmor MUSCL finite volume method \citep{KURGANOV2000241} and a second-order finite difference discretization of our chosen differential operator. The method is compared to bulk variables computed from the microscopic formulation for validation.   
% Guys at Brown
%1. Using a singular kernel makes the nonlocal alignment term to be a nonlinear function
%2. nonlocality: using the fast Fourier transform, resulting in a computation cost reduced to O(K log K).
%3. nonlinearity: we use a piecewise constant approximation, i.e., a finite volume scheme,
%which can significantly simplify the computation of the nonlinear nonlocal term. Moreover, we show
%that the proposed finite volume scheme preserves the mass and momentum, both in one and two space
%dimensions.

The rest of the manuscript is organized as follows: 
Section \ref{Sec:CS} introduces the agent-based Cucker-Smale flocking dynamics and 
the macroscopic Euler equations.
Section \ref{Sec:BVP} describes the conversion of the Euler equations 
to an augmented system of PDEs, and the formulation of the boundary value problem.
In Section \ref{Sec:Computations} a family of interaction functions is proposed and 
the computation process is explained.
Finally, Section \ref{Sec:Results} compares the numerical results and 
Section \ref{Sec:Conclusion} concludes the paper.

%%%%%%%%%%%%%%%%%%%%%%%%%%%%%%%%%%%%%%%%%%%%%%%%%%%%%%%%%%%%%%%%%%%%%%%%%%%%%%%%

%%%%%%%%%%%%%%%%%%%%%%%%%%%%%%%%%%%%%%%%%%%%%%%%%%%%%%%%%%%%%%%%%%%%%%%%%%%%%%%%
\section{Mathematical Models}
\label{Sec:CS}

In this section we introduce the Cucker-Smale dynamics under 
general interaction functions, define time-asymptotic flocking,
and present the mean-field macroscopic equations.

\subsection{The Cucker-Smale Model}

Consider an interacting system of $N$ identical autonomous agents
with unit mass in $\mathbb{R}^d$, $d \in \{1,2,3\}$. 
Let $x_i(t),\ v_i(t)\in\mathbb{R}^d$ represent the position and velocity 
of the $i^{th}$-particle at each time $t\geq 0$, respectively, for $1\leq i\leq N$.
Then the general Cucker-Smale dynamical system \citep{cucker2007emergent} of $(2N)$ ODEs
reads as:
\begin{equation}
\begin{cases}
\oder{x_i}{t} &= v_i \\
\oder{v_i}{t} &= \frac{1}{N}\sum_{j=1}^{N}\psi(x_j,x_i)(v_j-v_i) 
\end{cases}
\label{eq:cs}
\end{equation}
where $x_i(0)$, are $v_i(0)$ are given for all $i = 1,\ldots,N$, 
and ${\psi:\mathbb{R}^d \times \mathbb{R}^d \rightarrow \mathbb{R}}$ 
represents the interaction function between each pair of particles.
%a non-negative and non-increasing function. 
%as described in Section \ref{Sec:Preliminaries}, 
%and $x$, $v$ are defined over a finite support $x\in\Omega\subset\mathbb{R}^d$.

The center of mass system $(x_c,v_c)$ of $\cbra{(x_i,v_i)}_{i=1}^N$
is defined as
\begin{equation}
x_c = \frac 1 N \sum_{i=1}^N x_i, \quad v_c = \frac 1 N \sum_{i=1}^N v_i
\end{equation}
When $\psi$ is symmetric, i.e., $\psi(x,s)=\psi(s,x)$, 
system (\ref{eq:cs}) implies
\begin{equation}
\oder{x_c}{t}=v_c,\quad \oder{v_c}{t}=0
\end{equation}
which gives the explicit solution
\begin{equation}
x_c(t) = x_c(0) + t v_c(0),\ t\geq 0
\label{eq:center_of_mass}
\end{equation}

\subsection{Asymptotic Flocking}
\label{sSec:flocking}

We investigate the additional assumptions on the initial conditions and 
the interaction function $\psi$, such that
system (\ref{eq:cs}) converges to a velocity consensus, 
a phenomenon known in the literature as \emph{time-asymptotic flocking},
defined in terms of the center of mass system as
\begin{definition}[Asymptotic Flocking]
An $N-$body interacting system $\mathcal{G}=\cbra{(x_i,v_i)}_{i=1}^N$ 
% and $(x_c, v_c) := (\frac 1 N \sum_{i=1}^N x_i, \frac 1 N \sum_{i=1}^N v_i)$
% be the corresponding center of mass system.
exhibits time-asymptotic flocking if and only if the following two relations hold:
\begin{itemize}
\item(Velocity alignment): % The velocity fluctuations approach zero asymptotically, i.e.
%\begin{align}
$\lim_{t\rightarrow\infty}\sum_{i=1}^N \norm{v_i(t)-v_c(t)}^2=0$ ,
%\end{align}
\item(Spatial coherence): %The position fluctuations are uniformly bounded, i.e.
%\begin{align}
$\sup_{0\leq t \leq\infty} \sum_{i=1}^N \norm{x_i(t)-x_c(t)}^2 <\infty$ .
%\end{align}
\end{itemize}
\label{def:flocking}
\end{definition}

We consider the new variables 
\begin{equation}
(\hat x_i, \hat v_i):=(x_i-x_c, v_i-v_c)
\label{eq:fluctuations}
\end{equation}
which correspond to the fluctuations around the center of mass system,
and define
$\hat x:=(\hat x_1,\ldots,\hat x_N)$, $\hat v:=(\hat v_1,\ldots,\hat v_N)$, %\in\mathbb{R}^{Nd}$
$|\hat x| = \pbra{\sum_{i=1}^N \|\hat x_i\|^2}^{1/2}$, and 
$|\hat v| = \pbra{\sum_{i=1}^N \|\hat v_i\|^2}^{1/2}$,
where $\|\cdot\|$ represents the standard $l_2$-norm 
in $\mathbb{R}^d$.
Based on Definition \ref{def:flocking},
asymptotic flocking is achieved if 
\begin{equation}
|\hat x(t)|<\infty, t\geq 0,\ \text{and }
\lim_{t\rightarrow\infty} |\hat v(t)|=0
\end{equation}
We first notice that 
\begin{equation}
\oder{|\hat x|^2}{t}=2\abra{\oder{\hat x}{t},\hat x}\leq 2|\hat x||\hat v|
\end{equation}
which implies 
\begin{equation}
\oder{|\hat x|}{t} \leq |\hat v|
\label{eq:bound_x}
\end{equation}
Suppose the interaction function $\psi$ is chosen such that 
$\psi(x,s)=\tilde\psi(\|x-s\|)$, 
with $\tilde\psi:\mathbb{R}_+\rightarrow\mathbb{R}_+$ being a
non-negative and non-increasing function.
Then $(\hat x_i, \hat v_i)$ are governed by the dynamical system (\ref{eq:cs}), and
\begin{equation}
\begin{aligned}
\oder{|\hat v|^2}{t} &=-\frac 1 N \sum_{1\leq i,j\leq N} \tilde\psi(\|\hat x_j-\hat x_i\|)\|\hat v_j-\hat v_i\|^2 \\
&\leq -\frac 1 N \tilde\psi(2|\hat x|) \sum_{1\leq i,j\leq N} \|\hat v_j-\hat v_i\|^2 \\
&= -\frac 2 N \tilde\psi(2|\hat x|) |\hat v|^2 
\end{aligned}
\end{equation}
which implies
\begin{equation}
\oder{|\hat v|}{t} \leq -\frac 2 N \tilde\psi(2|\hat x|) |\hat v| :=-\phi(|\hat x|)|\hat v|
\label{eq:bound_v}
\end{equation}
where we have used the fact that  
$\sum_{i=1}^N \hat v_i(t)=0$, $t\geq 0$, and
\begin{equation}
\max_{1\leq i,j\leq N}\|\hat x_i-\hat x_j\|\leq 2|\hat x|
\label{eq:xixj_bound}
\end{equation}
The following Theorem by \citep{ha2009simple} provides 
sufficient conditions for time-asymptotic flocking.
\begin{theorem}
Suppose $(|x|,|v|)$ satisfy the system of dissipative differential inequalities 
(\ref{eq:bound_x}), (\ref{eq:bound_v}) with $\phi\geq 0$. 
Then if $|v(0)|<\int_{|x(0)|}^\infty \phi(s) ds $,
there is a $x_M\geq 0$ such that
$|v(0)|=\int_{|x(0)|}^{x_M} \phi(s) ds$, 
and for every $t\geq 0$, $|x(t)|\leq x_M$, and 
$|v(t)|\leq |v(0)|e^{-\phi(x_M)t}$.
\label{thm:flocking}
\end{theorem}
The following is an immediate consequence of Theorem \ref{thm:flocking}.
\begin{proposition}
Let $\mathcal{G}=\cbra{(x_i,v_i)}_{i=1}^N$ be an $N-$body interacting system
with dynamics given by (\ref{eq:cs}). 
% and the transformation $(\hat x_i, \hat v_i)$ as defined in (\ref{eq:fluctuations}).
Suppose %the interaction function $\psi$ is chosen such that 
$\psi(x,s)=\tilde\psi(\|x-s\|)$, 
with $\tilde\psi:\mathbb{R}_+\rightarrow\mathbb{R}_+$ being a
non-negative and non-increasing function.
Then if $|v(0)-v_c(0)|<\int_{|x(0)-x_c(0)|}^\infty \frac 2 N \tilde\psi(2s) ds $,
$\mathcal{G}$ exhibits time-asymptotic flocking.
% there is a $\hat x_M\geq 0$ such that
% $|\hat v(0)|=\int_{|\hat x(0)|}^{\hat x_M} \frac 2 N \tilde\psi(2s) ds$, 
% $|\hat x(t)|\leq \hat x_M$, and 
% $|\hat v(t)|\leq |\hat v(0)|e^{-\phi(\hat x_M)t}$.
\label{pro:flocking}
\end{proposition}
%
% Now we can show that the Lyapunov function
% %
% \begin{equation}
% V(|x|,|v|):=|\hat v|+ \int_{\alpha}^{|\hat x|} \phi(s) ds,\ \alpha\geq 0 
% \end{equation}
% %
% is non-increasing along the solutions of $(|\hat x(t)|,|\hat v(t)|)$ 
% of the system of dissipative differential inequalities 
% (\ref{eq:bound_x}) and (\ref{eq:bound_v}), since
% \begin{equation}
% \begin{aligned}
% \oder{}{t}V(|\hat x|,|\hat v|) & = \oder{|\hat v|}{t} + \phi(|\hat x|)\oder{|\hat x|}{t} \\
% &\leq \phi(|\hat x|) 
% \pbra{ -|v| + \oder{|\hat x|}{t}} \\
% &\leq 0
% \end{aligned}
% \end{equation}
% %
% which implies that 
% %
% \begin{equation}
% |\hat v(t)|+ \int_{|\hat x_0|}^{|\hat x|} \phi(s) ds 
% \leq |\hat v(0)|
% \label{eq:contradiction}
% \end{equation}
% %

% We have assumed that $|\hat v(0)|<\int_{|\hat x(0)|}^\infty \phi(s)ds $,
% such that, because of the non-negativity of $\phi$, 
% there exists a $\hat x_M\geq 0$ for which 
% %
% \begin{equation}
% |\hat v(0)|=\int_{|\hat x(0)|}^{\hat x_M} \phi(s)ds 
% \end{equation} 
% %

% Suppose there exists a $t^*\geq 0$, such that $\hat x^*:=|\hat x(t^*)|\geq \hat x_M$.
% Then 
% %
% \begin{equation}
% \int_{|\hat x(0)|}^{\hat x^*} \phi(s)ds > |v(0)| 
% \end{equation}
% %
% which contradicts (\ref{eq:contradiction}).
% %
% Therefore 
% \begin{equation}
% |\hat x(t)| \leq \hat x_M <\infty 
% \end{equation}
% %
% and from (\ref{eq:bound_v}) and the Grönwall-Bellman inequality we get
% %
% \begin{equation}
% |\hat v(t)| \leq |\hat v(0)|e^{-\phi(\hat x_M)t}
% \end{equation}
% %
% which completes the proof.
% \end{proof}

\subsection{The Mean-Field Limit}

%\begin{itemize}
%\item $f(t,x,v)$ to be the density of the particles in position $x$ with velocity $v$ at time $t$,
%\item $\rho(t,x) := \int_{\mathbb{R}^d} f \, dv$ to be the density of the particles in position $x$ at time $t$, and 
%\item $u(t,x) := \frac{\int_{\mathbb{R}^d} v f \, dv}{\int_{\mathbb{R}^d} f \, dv} = \frac{m(x,t)}{\rho(x,t)}$ to be the %bulk velocity in position $x$ at time $t$.  
%\end{itemize}
%
%Taking the mean-field limit \citep{carrillo2010particle} 
%the CS system (\ref{eq:cs}) results to the 
%Euler equations with commutator forcing:
%
Consider the empirical joint probability distribution of the particle positions and velocities $\{x_i,v_i\}_{i=1}^N$
\begin{equation}
    F_{xv}^N(t,x,v) := \frac{1}{N}\sum_{i=1}^N \delta(x_i,v_i)
\end{equation}
where $\delta(\cdot, \cdot)$ is the Dirac measure on $\mathbb R^{2d}$. 
As the number of particles $N \rightarrow \infty$, we can use
McKean-Vlasov arguments to show that the empirical distribution
converges weakly to a distribution whose density $f_{xv}$ evolves according to the
forward Kolmogorov
equation \citep{carrillo2010particle}
\begin{equation}\label{fpke}
\begin{split}
    &\partial_t f_{xv} + \nabla_x\cdot (vf_{xv})+\nabla_v\cdot (Af_{xv}) = 0\\
    &A := \int_{\mathbb{R}^{2d}}\psi(x,s)(w-v)f_{xv}(t,s,w)dsdw.
    \end{split}
\end{equation}

We define
\begin{equation}
    \begin{split}
        \rho(t,x) &:= \int_{\mathbb R^d} f_{xv}(t,x,v) dv \\
        m(t,x)&:=\rho(t,x)u(t,x) := \int_{\mathbb R^d} v f_{xv}(t,x,v) dv.
    \end{split}
\end{equation}

which are the marginal probability and momentum density functions. 
Substituting these into (\ref{fpke}) yields the following $(d+1)$ compressible Euler
equations with non-local forcing:
\begin{equation}
\begin{cases}
\partial_t{\rho} + \nabla_x \cdot (\rho u ) = 0 \\
\partial_t{(\rho u)} + \nabla_x \cdot (\rho u \otimes u) = 
\rho \mathcal L_\psi (\rho u) - \rho u \mathcal L_\psi \rho
\end{cases}
\label{eq:euler}
\end{equation}
where $u$ is the mean velocity, $\rho(0,x)$ and $u(0,x)$ are given and 
\begin{equation}
\mathcal L_\psi f  (t,x):= \int_{\mathbb{R}^d} \psi(x,s)f(t,s) ds.
\label{eq:convolution}
\end{equation}
%

%%%%%%%%%%%%%%%%%%%%%%%%%%%%%%%%%%%%%%%%%%%%%%%%%%%%%%%%%%%%%%%%%%%%%%%%%%%%%%%%

\section{Semi-linear Poisson Mediated Flocking}
\label{Sec:BVP}
%We describe a procedure of converting the integral
%terms into an equivalent and cheaper PDEs.

%The Euler equations (\ref{eq:euler}) 
%are challenging to be numerically solved, due to the integral form of 
%$\mathcal L_\psi$ in (\ref{eq:convolution}). %\red{(++ explain - Christos)}
% When dealing with integro-differential equations, 
% one typically tries to convert them 
% into differential equations, a process that involves some kind of differentiation on both sides. 

\subsection{Conversion to a system of PDEs}

We think of the function $\psi$ as a Green's function, i.e., 
as the impulse response of a linear differential equation, 
represented by the operator $\mathcal{L}_x$,
such that
%
%\begin{equation}
%\mathcal{L}_x y(t,x) = f(t,x)\implies y(t,x) = %\int_{\mathbb{R}^d} \psi(x,s) f(t,s) ds\label{eq:greens}
%\end{equation}
%
% \begin{equation}
% 	\begin{split}
% \mathcal{L}_x y(t,x) &= f(t,x) \\
% &\implies y(t,x) = \int_{\mathbb{R}^d} \psi(x,s) f(t,s) ds\label{eq:greens} \\
% &\implies \mathcal L^{-1}_\psi = \mathcal L_x
% 	\end{split}
% \end{equation}
% for all $t\geq 0 $, where 
%
\begin{equation}
    \mathcal{L}_x y(t,x) = g(t,x)
\end{equation}
implies
\begin{equation}
    y(t,x) = \int_{\mathbb{R}^d} \psi(x,s) g(t,s) ds
    \label{eq:greens}
\end{equation}
which results in 
\begin{equation}
    \mathcal L^{-1}_\psi = \mathcal L_x
\end{equation}
for all $t\geq 0 $, where 
%
%\begin{align}
%\mathcal{L}_s\psi(x,s) = \delta(x-s),\ x,s\in\mathbb{R}^d
%\end{align}
%
\begin{align}
\mathcal{L}_x\psi(x,s) = \delta(x-s),\ x,s\in\mathbb{R}^d.
\end{align}
Then the following proposition holds:
\begin{proposition}
Suppose $\psi$ is a Green's function 
with respect to a linear differential operator $\mathcal{L}_x$.
Then system (\ref{eq:euler}) %of $(d+1)$ integro-differential equations 
is equivalent to the augmented system 
of ($2d+2$) partial differential equations:
\begin{equation}
\begin{cases}
\partial_t{\rho} + \nabla_x \cdot (\rho u) = 0 \\
\mathcal{L}_x y = \begin{bmatrix} \rho u & \rho \end{bmatrix}^T\\
\partial_t{(\rho u)} + \nabla_x \cdot (\rho u \otimes u) = 
\sum_{i=1}^{d} (\rho y_i - \rho u_i y_{d+1}) \cdot \hat e_i
\end{cases}
\label{eq:pdes}
\end{equation}
where $\cbra{\hat e_i}_{i=1}^d$ is the standard basis in $\mathbb{R}^d$.
%and corresponds to a system of interacting particles,
%with dynamics given by (\ref{eq:cs}),
%that exhibit time-asymptotic flocking.
\end{proposition}
%

% \begin{proposition}
% Suppose there exists an interaction function 
% $\psi$, satisfying the conditions of Proposition \ref{pro:flocking}, 
% and being a Green's function 
% with respect to a linear differential operator $\mathcal{L}_x$.
% Then system (\ref{eq:euler}) %of $(d+1)$ integro-differential equations 
% is equivalent to the augmented system 
% of ($2d+2$) partial differential equations:
% %
% \begin{equation}
% \begin{cases}
% \partial_t{\rho} + \nabla_x \cdot (\rho u) = 0 \\
% \mathcal{L}_x y = \begin{bmatrix} \rho & \rho u\end{bmatrix}^T\\
% \partial_t{(\rho u)} + \nabla_x \cdot (\rho u \otimes u) = 
% \sum_{i=1}^d (\rho y_i - \rho u_i y_0) \cdot \hat e_i
% \end{cases}
% \label{eq:pdes}
% \end{equation}
% %
% where $\cbra{\hat e_i}_{i=1}^d$ is the standard basis in $\mathbb{R}^d$.
% %and corresponds to a system of interacting particles,
% %with dynamics given by (\ref{eq:cs}),
% %that exhibit time-asymptotic flocking.
% \end{proposition}
%
\subsection{The Boundary Value Problem}

% System (\ref{eq:pdes}), 
% can only be numerically solved as a Boundary Value Problem (BVP)
% in a bounded open domain $\Omega\subset\mathbb{R}^d$. 
Due to the time-dependence of the center of mass (\ref{eq:center_of_mass}),
$x_i$, $i=1,\ldots,N$, will escape any fixed and open bounded domain $\Omega\subset\mathbb{R}^d$, 
unless in the trivial case where $v_c(0)=0$.
Because of the flocking behavior (Definition \ref{def:flocking}),
the position fluctuations with respect to the center of mass are uniformly bounded, i.e.,
\begin{equation}
    \sup_{0\leq t \leq\infty} \sum_{i=1}^N \norm{x_i(t)-x_c(t)}^2 <\infty
\end{equation}
and, therefore we can define a Boundary Value Problem (BVP) in the moving domain 
\begin{equation}
    \Omega_c(t)=\cbra{x+x_c(t):x\in\Omega}
\end{equation}
where it is assumed that $0_d\in\Omega$, $0_d$ being the origin of $\mathbb{R}^d$.

%Solving a BVP in a moving domain can be extremely challenging 
%\citep{cortez2013pdes}, if even possible. 
We notice that solving system (\ref{eq:pdes}) 
for $(x,u)$, $x\in\Omega_c$ is equivalent to solving it 
for the fluctuation variables $(\hat x,\hat u)$ (\ref{eq:fluctuations}), 
with $\hat x\in\Omega$.
%
% With regard to the below, we did not prove existence, uniqueness, and continuous
% dependence on the data, so I'd rather not make the bold claim of
% well-posedness
%
%
%
%Moreover, given a bounded open domain $\Omega\subset\mathbb{R}^d$,
%system (\ref{eq:pdes}), along with the initial conditions
%\begin{equation}
%\rho(0,\hat x) = \rho_0(\hat x),\
%v(0,\hat x)=v_0(\hat x) 
%\label{eq:bvp_pdes}
%\end{equation}
%
%and the boundary conditions
%
%\begin{equation}
%\rho(t,\partial\Omega) = 0,\ y_i(t,\partial\Omega)=0,\ i=1,\ldots,d+1
%\label{eq:bvp_pdes_y}
%\end{equation}
%
%where the functions $\rho_0$ and $u_0$ are given,
%is a well-posed BVP.

We note that the boundedness of the domain 
has an effect on both the Green's function
and the flocking behavior of the system of interacting particles,
which should satisfy 
\begin{equation}
    x_i(t)-x_c(t) \in\Omega,\ i=1,\ldots,N,\ t\geq 0 .
\end{equation}

% \subsection{The effect of the Bounded Domain on the Green's Function}

% The Boundary Value Problem
% %
% %\begin{equation}
% %\begin{aligned}
% %\mathcal{L}_x y(t,x) &= f(t,x),\ x\in\Omega \\
% %y(t,x)&=0,\ x\in\partial\Omega
% %\end{aligned} 
% %\label{eq:greens_bvp}
% %\end{equation}
% %
% %for all $t$, admits a Green's function $\hat\psi$ such that 
% %
% \begin{equation}
% \begin{aligned}
% \mathcal{L}_s \hat\psi(x,s) &= \delta(x-s),\ s\in\Omega \\
% s.t.\quad \hat\psi(x,s) &=0,\ s\in \partial\Omega
% \end{aligned}
% \end{equation}
% %
% admits a Green's function $\hat\psi$ of the form  
% %
% \begin{align}
% \hat\psi(x,s) = \psi(x,s) + \phi(x,s)
% \end{align}
% %
% where $\phi$ is a function satisfying
% %
% \begin{equation}
% \begin{aligned}
% \mathcal{L}_s \phi(x,s) &= 0,\ x\in\Omega \\
% s.t.\quad \phi(x,s) &= -\psi(x,s),\ x\in\partial\Omega  
% \end{aligned}
% \label{eq:corrector}
% \end{equation}
% %%
% %while $\hat\psi$ converges to $\psi$ in the limit $\Omega$, i.e.,
% %%
% %\begin{equation}
% %\lim_{\Omega\rightarrow\mathbb{R}^d} \hat\psi(x,s) = \tilde\psi(\|x-s\|)
% %\end{equation}
% %%

% The flocking behavior of the system of interacting particles
% will be re-considered in Section \ref{Sec:Computations}
% due to the effect of the bounded domain on the Green's function $\hat\psi$. 

%%%%%%%%%%%%%%%%%%%%%%%%%%%%%%%%%%%%%%%%%%%%%%%%%%%%%%%%%%%%%%%%%%%%%%%%%%%%%%%%

\section{One-Dimensional Case}
\label{Sec:Computations}

%Next we show the complete methodology for 1-D.
The BVP of the augmented system of PDEs (\ref{eq:pdes}) for $d=1$, 
on $\Omega = \{\hat x \in [-\frac{L}{2},\frac{L}{2}]\}$ reads as:
%\begin{equation}
%\begin{cases}
%\pder{\rho}{t} + \pder{(\rho v)}{x} = 0 \\
%\mathcal{L}_x y_0 = \rho \\
%\mathcal{L}_x y_1 = \rho v \\
%\pder{(\rho v)}{t} + \pder{(\rho v^2)}{x} = 
%	\rho y_1 - \rho v y_0
%\end{cases}
%\label{eq:pdes1}
%\end{equation}
%
\begin{equation}
\begin{cases}
\partial_t{\rho} + \partial_{\hat x}(\rho u) = 0 \\
\mathcal{L}_{\hat x} y = \begin{bmatrix} \rho u & \rho  \end{bmatrix}^T \\
\partial_t(\rho u) + \partial_{\hat{x}} (\rho u^2)= 
	\rho y_1 - \rho u y_2
\end{cases}
\label{eq:pdes1}
\end{equation}
with homogeneous Dirichlet boundary conditions and initial conditions
\begin{equation}
        \rho(0, \hat{x}) = \rho_0(\hat x), \quad
        u(0, \hat{x}) = u_0(\hat x) 
\label{eq:bvp_pdes}
\end{equation}
which are smooth functions.

We select the linear partial differential operator
\begin{equation}\label{op}
\mathcal{L}_x= -\frac{1}{2k}(\frac{\partial^2}{\partial x^2} - \lambda^2)
\end{equation}
with $k\neq 0$ and $\lambda \neq 0$, for which the associated parametric family of Green's functions with homogeneous Dirichlet boundary conditions on $[0,L]$ reads as:

\begin{equation}
\hat\psi(x,s) = 
\begin{cases}
c_1(s) ( e^{\lambda x} - e^{-\lambda x}) & s\leq x \\
c_2(s) ( e^{\lambda(x - 2L)} - e^{-\lambda x} ) & s>x 
\end{cases}
\label{eq:psi}
\end{equation}
\begin{equation}
\begin{aligned}
c_1(s) &=  \frac {k}{\lambda(e^{-2L\lambda}-1)} (e^{\lambda(s-2L)} - e^{-\lambda s}) \\
c_2(s) &= \frac {k}{\lambda(e^{-2L\lambda}-1)}(e^{\lambda s} - e^{-\lambda s})
\end{aligned}
\label{eq:psi_coef}
\end{equation}
The solution over any interval of length $L$ can be obtained by a simple translation of coordinates.

%\begin{equation}
%\hat\psi(x,s) = 
%\begin{cases}
%a(x) ( e^{-\lambda s} - e^{\lambda L} e^{\lambda s}),\ s\leq x \\
%c(x) ( e^{-\lambda s} - e^{-\lambda L} e^{\lambda s}),\ s>x 
%\end{cases}
%\label{eq:psi}
%\end{equation}
%\begin{equation}
%\begin{aligned}
%a(x) =  - \frac k \lambda \frac{1}{e^{\lambda L}-e^{-\lambda L}}(e^{-\lambda x} - e^{-\lambda L} e^{\lambda %x})\\
%c(x) = - \frac k \lambda \frac{1}{e^{\lambda L}-e^{-\lambda L}}(e^{-\lambda x} - e^{\lambda L} e^{\lambda x})
%\end{aligned}
%\label{eq:psi_coef}
%\end{equation}

The profile of the Green's function $\hat \psi$ and the effect of the bounded domain on on it is 
illustrated in Fig. \ref{fig:psi}, where, 
for different fixed values of $x$,
$\hat\psi(x,s)$ is compared to the function 
\begin{equation}
  \psi(x,s) = \frac{k}{\lambda}e^{-\lambda\|x-s\|}  
\end{equation}
which is the Green's function corresponding to $\mathcal{L}_x$ in an infinite domain.
We note that the parameters $(k,\lambda,L)$ generate a family of interaction functions 
(see also \citep{mavridis2020learning})
that can simulate widely used interaction functions as the one found in the original Cucker-Smale model
\citep{cucker2007emergent}:
\begin{equation}
G(x,s) = \frac{K}{(1+\|x-s\|^2)^{\gamma}}    
\end{equation}
for given parameters $(K,\gamma)$.
% The expressiveness of the parametric family of interaction functions $\hat\psi$ 
% % for different parameter values $k$ and $\lambda$ 
% is depicted in Fig.\ref{fig:psi_param}.
 
\begin{figure}[ht!]
        \centering
        \includegraphics[trim=5 5 5 40, clip,width=.4\textwidth]{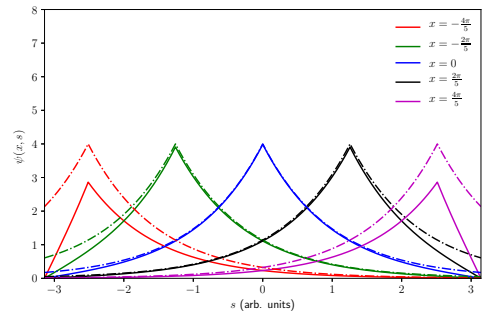} 
        % \vspace{-2em}
        \caption{Illustration of $\hat\psi(x,\cdot)$ (\ref{eq:psi}) for different values of $x$, and for 
        $\lambda=1$, $k=4$ on $[-\pi, \pi]$. 
        % The Green's function is equivalent to zero for $s = \pm \pi$.
        The function $ \psi(x,s)= \frac{k}{\lambda}e^{-\lambda\|x-s\|}$, which is the 
        Green's function for $\mathcal{L}_x$ in infinite domain, is depicted in the dashed-dotted lines.}
        \label{fig:psi}
\end{figure}

%\begin{figure}[ht!]
        %\centering
        %\includegraphics[trim=5 5 5 40, clip,width=.5\textwidth]{figures/Kernel_params.eps}
        %\vspace{-2em}
        %\caption{The effect of the parameters $k$, and $\lambda$ on the 
        %profile of the interaction function $\hat\psi$.The function $ \psi(x-s)= %\frac{k}{\lambda}e^{-\lambda\|x-s\|}$, which is the 
        %Green's function for $\mathcal{L}_x$ in infinite domain, is depicted in the dashed %lines. Here, $x=0$.}
%        \label{fig:psi_param}
%\end{figure}

% \newline

\subsection{Asymptotic Flocking}

Next we provide sufficient conditions such that the solution $\cbra{(x_i(t),v_i(t))}_{i=1}^N$, $t\geq 0$, 
of system (\ref{eq:cs}) with interaction function $\hat\psi$ 
as defined in (\ref{eq:psi}), (\ref{eq:psi_coef}), satisfy the 
flocking conditions in Definition \ref{def:flocking}, with 
$\hat x_i(t)\in\Omega$, for all $t\geq 0$.

%\begin{theorem}
%Let $\mathcal{G}=\cbra{(x_i,v_i)}_{i=1}^N$ be an $N-$body interacting system
%with dynamics given by (\ref{eq:cs}), and 
%the transformation $(\hat x_i, \hat v_i)$ as defined in (\ref{eq:fluctuations}).
%Suppose %the interaction function $\psi$ is chosen such that 
%$\psi(x,s)=\hat\psi(x,s)$, as defined in (\ref{eq:psi}), (\ref{eq:psi_coef})
%for some $L>0$.
%Then if $|\hat v(0)|<\int_{|x(0)|}^{\hat x_M} \frac 2 N \tilde\psi(2s) ds $,
%there is a $\hat x_M\geq 0$ such that
%$|\hat v(0)|=\int_{|\hat x(0)|}^{\hat x_M} \frac 2 N \tilde\psi(2s) ds$, 
%$|\hat x(t)|\leq \hat x_M$, and 
%$|\hat v(t)|\leq |\hat v(0)|e^{-\phi(\hat x_M)t}$.
%\label{thm:flocking}
%\end{theorem}
%%
%\begin{proof}

Similar to Section \ref{sSec:flocking}, 
we notice that 
%%
%\begin{equation}
%\oder{|\hat x|^2}{t}=2\abra{\oder{\hat x}{t},\hat x}\leq 2|\hat x||\hat v|
%\end{equation}
%%
%which implies 
\begin{equation}
\oder{|\hat x|}{t} \leq |\hat v|
\label{eq:bound_x1}
\end{equation}
From (\ref{eq:xixj_bound}) and the fact that 
$\|\hat x_i\|\leq \max_{1\leq i,j\leq N}\|\hat x_i-\hat x_j\|$, we get 
\begin{equation}
|\hat x|\leq \frac{\hat x_M}{2} \implies \|\hat x_i\| \leq \hat x_M,\ i=1,\ldots,N
\end{equation}
Therefore, we are interested in showing asymptotic flocking with 
$|\hat x(t)|\in[0,\frac{\hat x_M}{2}]$, for all $t\geq 0$.

For any given initial conditions $|\hat x_i(0)|$, 
there is a large enough value of $L$ such that there exist an $\hat x_M\in[0,\frac L 2)$
for which 
\begin{equation}
%\hat x_M>\max_{i=1,\ldots,N} \|\hat x_i(0)\|
\hat x_M > 2|\hat x(0)|
\end{equation}

From (\ref{eq:psi}), (\ref{eq:psi_coef}) it follows that
for $|\hat x| \leq \frac{\hat x_M}{2}$, 
\begin{equation}
\begin{aligned}
\hat\psi(x_j,x_i) &\geq \hat\psi(-\hat x_M, \|\hat x_j-\hat x_i\|) \\
&\geq \hat\psi(-\hat x_M, 2|\hat x|)
\end{aligned}
\end{equation}
which implies that
\begin{equation}
\begin{aligned}
\oder{|\hat v|^2}{t} &=-\frac 1 N \sum_{1\leq i,j\leq N} \hat\psi(\hat x_j, \hat x_i)\|\hat v_j-\hat v_i\|^2 \\
&\leq -\frac 2 N \hat\psi(-\hat x_M, 2|\hat x|) |\hat v|^2 
\end{aligned}
\end{equation}
and 
\begin{equation}
\oder{|\hat v|}{t} \leq -\frac 2 N \hat\psi(-\hat x_M, 2|\hat x|) |\hat v| :=-\phi(|\hat x|)|\hat v|
\label{eq:bound_v1}
\end{equation}

Next we notice that the Lyapunov function
\begin{equation}
V(|x|,|v|):=|\hat v|+ \int_{\alpha}^{|\hat x|} \phi(s) ds,\ \alpha\geq 0 
\end{equation}
is non-increasing along the solutions of $(|\hat x(t)|,|\hat v(t)|)$ 
of the system of dissipative differential inequalities 
(\ref{eq:bound_x}) and (\ref{eq:bound_v}), 
for $|\hat x(t)| \leq \frac{\hat x_M}{2}$, since
\begin{equation}
\begin{aligned}
\oder{}{t}V(|\hat x|,|\hat v|) & = \oder{|\hat v|}{t} + \phi(|\hat x|)\oder{|\hat x|}{t} \\
&\leq \phi(|\hat x|) 
\pbra{ -|v| + \oder{|\hat x|}{t}} \\
&\leq 0
\end{aligned}
\end{equation}
which implies that 
\begin{equation}
|\hat v(t)|+ \int_{|\hat x_0|}^{|\hat x|} \phi(s) ds 
\leq |\hat v(0)|,\ |\hat x| \leq \frac{\hat x_M}{2}
\label{eq:contradiction1}
\end{equation}

Choosing the initial velocity $|\hat v(0)|$ 
such that $|\hat v(0)|<\int_{|\hat x(0)|}^{\hat x_M/2} \phi(s)ds $,
and, since $\phi$ is non-negative for $|\hat x(t)| \leq \frac{\hat x_M}{2}$, 
there exists a $\bar x \in [|\hat x(0)|,\frac{\hat x_M}{2}]$ for which 
\begin{equation}
|\hat v(0)|=\int_{|\hat x(0)|}^{\bar x} \phi(s)ds 
\end{equation} 

Suppose there exists a $t^*\geq 0$, such that $\hat x^*:=|\hat x(t^*)|\in (\bar x,\frac{\hat x_M}{2}]$.
Then 
\begin{equation}
\int_{|\hat x(0)|}^{\hat x^*} \phi(s)ds > |v(0)| 
\end{equation}
which contradicts (\ref{eq:contradiction1}).
Therefore 
\begin{equation}
|\hat x(t)| \leq \bar x \leq \frac{\hat x_M}{2},\ t\geq 0  
\end{equation}
and from (\ref{eq:bound_v}) and the Grönwall-Bellman inequality 
\begin{equation}
|\hat v(t)| \leq |\hat v(0)|e^{-\phi(\bar x)t},\ t\geq 0.
\end{equation}
\subsection{Conservation of Mass and Momentum}
\begin{lemma} The operator (\ref{op}) $ 
\mathcal L_x$ on $C^{\infty}_{\mathbb R, C}(\Omega)$, the space
of compactly supported test functions, is self-adjoint
and invertible, and therefore has a self-adjoint inverse 
$\mathcal L_x^{-1}$ on $C^{\infty}_{\mathbb R, C}(\Omega)$.
\end{lemma}
\begin{proof} Self-adjointness of the inverse follows immediately from self-adjointness of $\mathcal L_x$ and the existence of the inverse \citep{taylor2010partial}. It is clear that $\mathcal L_x$ has an inverse since the Green's function is nontrivial.

We shall now show that the operator $\mathcal L_x$ is self-adjoint on $C^{\infty}_{\mathbb R, C}(\Omega)$.
Consider two functions $u,w \in C^{\infty}_{\mathbb R, C}(\Omega)$, $u\neq w$, 
the space
of test functions, and associated
$f_u, f_w \in C^{\infty}_{\mathbb R, C}$, $f_u := \mathcal L_x u, f_w := \mathcal L_x w$.
Let $\Omega := [-\frac{L}{2},\frac{L}{2}]$.
We have
\begin{equation}
    \int_\Omega (w \mathcal L_x u - u \mathcal L_x w)dx = 
    -\frac{1}{2k}\int_\Omega (w \partial_x^2 u - u \partial_x^2 w)dx.
\end{equation}
since the semi-linear term drops out. Using Green's second identity, and the compact support of $u,w$, we have that
\begin{equation}
    \int_\Omega (w \partial_x^2 u - u \partial_x^2 w)dx=
    \int_{\partial \Omega}  (w \partial_{\mathbf n}u - u\partial_{\mathbf n}w)dx=0.
\end{equation}
Thus, $\mathcal L_x$ is self-adjoint and has a self-adjoint inverse, i.e.
\begin{equation}\label{conserved}
    \int_\Omega (f_w \mathcal L_x^{-1} f_u - f_u \mathcal L_x^{-1} f_w) dx=
    \int_{\Omega } (f_w u - f_u w)dx = 0. 
\end{equation}

\end{proof}
\begin{proposition}
If $y$ is compactly supported, and $psi$ is as given, then mass
and momentum are conserved, i.e.
\begin{equation}\label{cons2}
\frac{d}{dt}\int_\Omega \begin{bmatrix} \rho & \rho u \end{bmatrix}^T d\hat x = 
    \int_\Omega \begin{bmatrix} 0 & \rho y_1 - \rho u y_2  \end{bmatrix}^T d \hat x = 0.
\end{equation}
\end{proposition}
\begin{proof}We obtain (\ref{cons2}) by simply integrating the conservation
laws in (\ref{eq:pdes1}) over the entire space and apply the Leibniz rule.
The conclusion follows directly from the self
adjointness of the inverse in (\ref{conserved}). The proposition holds for any self-adjoint alignment operator.
\end{proof}

\subsection{Computational Methods}

%In Section \ref{Sec:results} we present a comparison between the numerical solutions of 
%systems (\ref{eq:cs}), (\ref{eq:1conv}), and (\ref{eq:1aux}).
For compactness, we re-write the PDEs (\ref{eq:pdes1}) as

\begin{equation}
\begin{cases}
\partial_t U + \partial_{\hat x} F(U) = S(U, Y) \\
\mathcal L_{\hat x} Y = U \\
\end{cases}
\label{eq:1aux}
\end{equation}
with $U = \begin{bmatrix}\rho, \rho u \end{bmatrix}^T$, $Y = \begin{bmatrix}y_2, y_1 \end{bmatrix}^T$, $F = \begin{bmatrix} \rho u , \rho u^2 \end{bmatrix}^T$,
and $S = \begin{bmatrix}0, \rho y_1 - \rho u y_2 \end{bmatrix}^T$. Recall the transformation $m = \rho u$. 
From this, the flux Jacobian is given by
\begin{equation}
    \mathbf D_U F :=\begin{bmatrix}
        0 & 1 \\ -u^2 & 2u
    \end{bmatrix}
\end{equation}
which is not diagonalizable, and thus the system is only weakly hyperbolic. Its eigenvalues are $\pm u$.
With these notations established, we now detail the numerical solution of the PDEs.
\subsubsection{Hyperbolic Solver.}
To solve the hyperbolic system, we apply the finite volume method \citep{leveque_2002}. To begin, we define the sequence of points $\{\hat x_0, ..., \hat x_i, ..., \hat x_N \}$
which are the centers of the cells 
$I_i := [\hat x_{i-\frac{1}{2}}, \hat x_{i+\frac{1}{2}})$. Then, we average the PDE
over these cells, which gives
\begin{equation}
    \frac{1}{\lambda(I_i)}\frac{d}{dt}\int_{I_i} U d\hat x = 
    -\frac{1}{\lambda(I_i)}\int_{I_i} \partial_{\hat x} F d\hat x + 
    \frac{1}{\lambda(I_i)}\int_{I_i} S d\hat x
\end{equation}
where $\lambda(\cdot)$ denotes the length of an interval. Suppose these are
identical, so $\Delta \hat x :=\lambda(I_i) \forall i$. Then, using the divergence
theorem, and replacing the integrals of $U,F, S$ with their cell-averages,\
i.e. their midpoint values $\bar U, \bar F, \bar S$, we obtain
\begin{equation}
    \frac{d}{dt} \bar U_i = 
    -\frac{1}{\Delta \hat x}(\bar F_{i+\frac{1}{2}} - \bar F_{i-\frac{1}{2}}) + 
    \bar S_i
\end{equation}

where $\bar U_i := \bar Y(\hat x_i), \bar F_i := \bar F(\hat x_i), \bar S := \bar S(\hat x_i)$.
In this work, we employ the second-order strong stability preserving Runge-Kutta scheme \citep{KURGANOV2000241} for time integration. For the fluxes, we assume piecewise linearity and use the Kurganov-Tadmor flux \citep{KURGANOV2000241}. The fluxes are given by
\begin{equation}
\begin{split}
    \bar F_{i+\frac{1}{2}}&:=\frac{1}{2}[F^*_{i} + F^*_{i+1} - \max\{|u^*_i|, |u^*_{i+1}|\}(U^*_{i+1}-U^*_{i})]\\
    U^*_{i+1} &:= U_{i+1} - \frac{\Delta \hat x}{2}minmod(\frac{U_{i+2} - U_{i+1}}{\Delta \hat x},\frac{U_{i+1} - U_{i}}{\Delta \hat x})\\
     U^*_{i} &:= U_{i} + \frac{\Delta \hat x}{2}minmod(\frac{U_{i+1} - U_{i}}{\Delta \hat x},\frac{U_{i} - U_{i-1}}{\Delta \hat x})\\
    \end{split}
\end{equation}
where $minmod(a,b) := \frac{1}{2}(sign(a)+sign(b))\min(|a|, |b|)$. 

\subsubsection{Elliptic Solver.}

To solve the elliptic equations, we apply the classical second-order finite difference method, which is
\begin{equation}
        \frac{y^j_{i+1} - 2y^j_i + y^j_{i-1}}{\Delta \hat x^2} - \lambda^2 y_i^j = -2kU^j_i
    \end{equation}
Over the interior points, this yields linear equations
\begin{equation}\label{linsys}
    (\frac{1}{\Delta \hat x^2}\mathbf A - \lambda^2 \mathbf I)y^j_{int} = -2kU^j_{int} - 
\frac{1}{\Delta \hat x^2}\begin{bmatrix} y^j_{0} & 0 & \hdots & 0 & y^j_{N}\end{bmatrix}^T,
\end{equation}
\begin{equation}
\mathbf A = \begin{bmatrix}-2 & 1 & 0 & \hdots & \hdots & 0 \\
1 &  -2 & 1 & 0& \hdots & 0 \\
0 &  1 & -2 & 1& \hdots & 0 \\
\vdots & \ddots & \ddots & \ddots & \ddots & \vdots \\
0 & 0 & 0 & 0 & 1 & -2\end{bmatrix}
\end{equation}
The matrix in (\ref{linsys}) is tridiagonal, so banded matrix algorithms \citep{golub} can be used to solve the corresponding system of equations. 
As shown in Fig. \ref{fig:speed}, using finite differences is much faster than a convolution (Riemann) sum, even when the embarrassing parallelism of the sum is exploited. 
\subsubsection{Particle Solver.}
We solve the system of particle equations using the velocity Verlet algorithm \citep{karniadakis_flockingDynamicsfractionalPDEs_2018}. Given 
a system of ODEs of the form
\begin{equation}
    \begin{cases}
    \frac{dx}{dt} &= v \\
    \frac{dv}{dt} &= a(x,v,t),
    \end{cases}
\end{equation}
with appropriate initial conditions 
and a time-discretization at steps $\{0,1,...,i,...\}$ with increment \
$\Delta t$, the discretization is
\begin{equation}
    \begin{split}
    v_{i+\frac{1}{2}} &= v_i + \frac{1}{2}a(x_i,v_i, t_i)\Delta t\\
    x_{i+1} &= x_i + \Delta t v_{i+\frac{1}{2}} \\
    v_{i+1} &= v_i + \frac{\Delta t}{2}[a(x_i,v_i,t_i) + a(x_{i+1}, v_{i+\frac{1}{2}}, t_{i+1})].
    \end{split}
\end{equation}
\begin{figure}[ht!]
        \centering
        \includegraphics[trim=0 0 0 40,clip,width=0.4\textwidth]{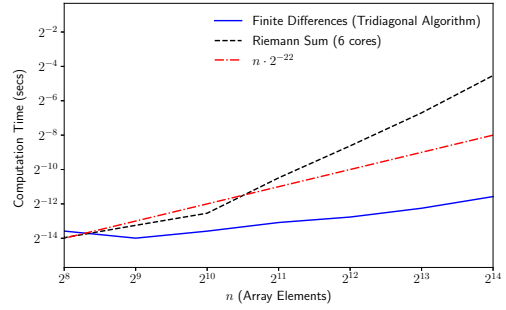}
        % \vspace{-2em}
        \caption{Computation Times for Nonlocal Terms using Finite Differences and Riemann Sum.}
        \label{fig:speed}
\end{figure}
%\subsection{Port-Hamiltonian Formulation}
%
%We can apply the theory of port-Hamiltonian formulation of infinite dimensional systems \cite{h31,h32}, in order to formulate the system (\ref{eq:bvp_aux}) as a port-hamiltonian system.
%
%First we define ${\bf{z}}^T = [{\bf p}^T,{\bf q}^T]$, with  ${\bf p}^T=[p_1,p_2,p_3]$ and ${\bf q}^T=[q_{12},q_{23},q_{31}]$, and the Hamiltonian function $H(z) = \ldots$.
%
%\red{TBD - Christos: Formulate Hamiltonian, follow \cite{h32}}

%%%%%%%%%%%%%%%%%%%%%%%%%%%%%%%%%%%%%%%%%%%%%%%%%%%%%%%%%%%%%%%%%%%%%%%%%%%%%%%%

%%%%%%%%%%%%%%%%%%%%%%%%%%%%%%%%%%%%%%%%%%%%%%%%%%%%%%%%%%%%%%%%%%%%%%%%%%%%%%%%
\section{Numerical Results and Higher Dimensions}
\label{Sec:Results}

In this section we present numerical simulations of one-dimensional nonlocal flocking dynamics, by solving $(a)$ the agent-based Cucker-Smale model using the velocity Verlet method, and 
$(b)$ the macroscopic model with initial conditions whose support is the
interval $[-\pi, \pi]$. Our aim is to verify that the agent based and continuum based approaches to the flocking problem produce similar results.

In the following, the initial density and velocity are given by
\begin{align}
\rho_0(\hat x) &= \frac{\pi}{2L} \cos(\frac{\pi \hat x}{L}), \\ 
u_0(\hat x) &= -c \sin(\frac{\pi \hat x)}{L}),\ \hat x\in[-\frac L 2,\frac L 2], 
\end{align}
i.e. it is assumed that $\rho_0(\hat x) = u_0(\hat x) = 0,\ \forall \hat x\notin [-\frac L 2,\frac L 2]$,
where we have used $L=2\pi$.

\subsection{Cucker-Smale Model Simulation}

In all simulations, we take $\lambda=1$, $k=4$.
For the particle simulation, we use $N=10^4$ particles. For the macro-scale
simulation, we use $\Delta \hat x = \frac{2\pi}{600}$ as the spatial increment. In
both simulations, we take $\Delta t = .001$ as the time increment.

In both cases, the support of the initial profile shrinks as the bulk comes
together.  The semi-linear Poisson-forced Euler system is highly dissipative, and the momentum profile
is damped until it flattens (although it is conserved over the domain), and the system attains an equilibrium
distribution.
Fig. \ref{fig:macros_micro} shows the agreement between
the particle model and the macro-scale model.
\begin{figure}[ht!]
        \centering
        \includegraphics[trim=10 20 10 80, clip,width=0.35\textwidth]{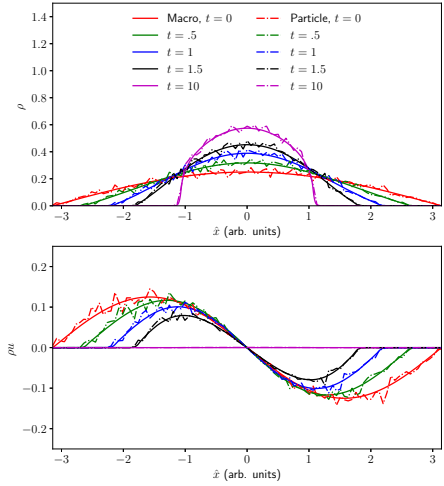}
        % \vspace{-2em}
        \caption{Evolution of the Probability Densities $\rho(t,\hat x)$ and Momentum Densities $m(t,\hat x)$ as computed by solving the macro-scale model and the particle model (dashed-line).}
        \label{fig:macros_micro}
\end{figure}

%\subsection{Comparison}
%\label{sSec:Comparison}
%Compare:
%\begin{itemize}
%\item Simulation vs Euler
%\item simulation vs Augmented
%\item Euler vs Augmented
%\item Simulation vs Euler w/ learned interaction
%\item Simulation vs Augmented w/ learned interaction 
%\end{itemize}

%\subsection{Comparison}
%\label{sSec:Comparison}
%Compare:
%\begin{itemize}
%\item Simulation vs Euler
%\item simulation vs Augmented
%\item Euler vs Augmented
%\item Simulation vs Euler w/ learned interaction
%\item Simulation vs Augmented w/ learned interaction 
%\end{itemize}

\subsection{Higher Dimensions}

In higher dimensions, the radial symmetry of the interaction function $\psi$ 
suggests the use of a \emph{singular kernel}. 
Singular kernels have been extensively studied in the literature and, under
mild assumptions in the initial conditions, have been shown to result in flocking behavior
while, at the same time, avoiding collisions \citep{Ahn2012OnCI}.%

In the BVP of the augmented system of PDEs (\ref{eq:pdes}) 
with the initial and boundary conditions (\ref{eq:bvp_pdes}), %(\ref{eq:bvp_pdes_y}),
we select the linear differential operator 
(see also \citep{mavridis2020learning}):
\begin{equation}
\mathcal{L}_x= - k^{-d/2} ( \nabla_x^2 - \lambda^2)
\end{equation}
and $\Omega = B_d(0,r):=\cbra{x\in\mathbb{R}^d:\|x\| < r}$,
which results in a Green's function of the form
\begin{align}
\hat\psi(x,s) = \psi(x-s) + \phi(x,s)
\end{align}
where $\psi$ is given by
\begin{equation}
\begin{aligned}
\psi(x,s) &= \tilde\psi(\|x-s\|) \\
&= \pbra{\frac{k}{2\pi}}^{d/2} \pbra{\frac{\lambda}{\|x-s\|}}^{d/2-1} 
K_{d/2-1}(\lambda \|x-s\|)
\end{aligned}
\end{equation}
with $K_\alpha(\cdot)$ being the modified Bessel function of the second kind of order $\alpha$, 
and $\phi$ is a function such that
\begin{equation}
\begin{aligned}
&\mathcal{L}_s \phi(x,s) = 0,\ s\in B_d(0,r) \\
&\phi(x,s) = -\psi(x,s),\ s\in\partial B_d(0,r)  
\end{aligned}
\end{equation}

For $s\in \partial B_d(0,r)$ we have

\begin{equation}
\begin{aligned}
\|x-s\|^2 &= \|x\|^2 - 2 \abra{x,s} + \|s\|^2 \\
%&= \|x\|^2 \frac{\|s\|^2}{r^2} - 2 \abra{x,s} + r^2 \\
%&= \|x\|^2 \pbra{ \frac{\|s\|^2}{r^2} - 2 \frac{\abra{x,s}}{\|x\|^2} + %\frac{r^2}{\|x\|^2} }\\
&= \|x\|^2 \| \frac{s}{r} - \frac{r x}{\|x\|^2}\|^2
\end{aligned}
\end{equation}
%
% and since for $x\in B_d(0,r)$, the dual point $x^*:=r^2 \frac{x}{\|x\|^2}\notin B_d(0,r)$,
and it can be shown that 
%  $\tilde\psi(\frac 1 r \|x\| \| s - x^*\|)$ 
%  is harmonic for all $s$ in $B_d(0,r)$, such that
%
\begin{equation}
\phi(x,s) = -\tilde\psi(\frac 1 r \|x\|\| s - r^2 \frac{x}{\|x\|^2}\|).
\end{equation}

The interaction function $\hat\psi$ 
is affected by the bounded domain in the same way as in the 
one-dimensional case, and depends on
the parameter values $k$ and 
$\lambda$ as illustrated in 
Fig.\ref{fig:psi_param_high} for the $2$-dimensional case.
 
\begin{figure}[b]
        \centering
        \includegraphics[trim=60 20 110 60,clip,width=0.17\textwidth]{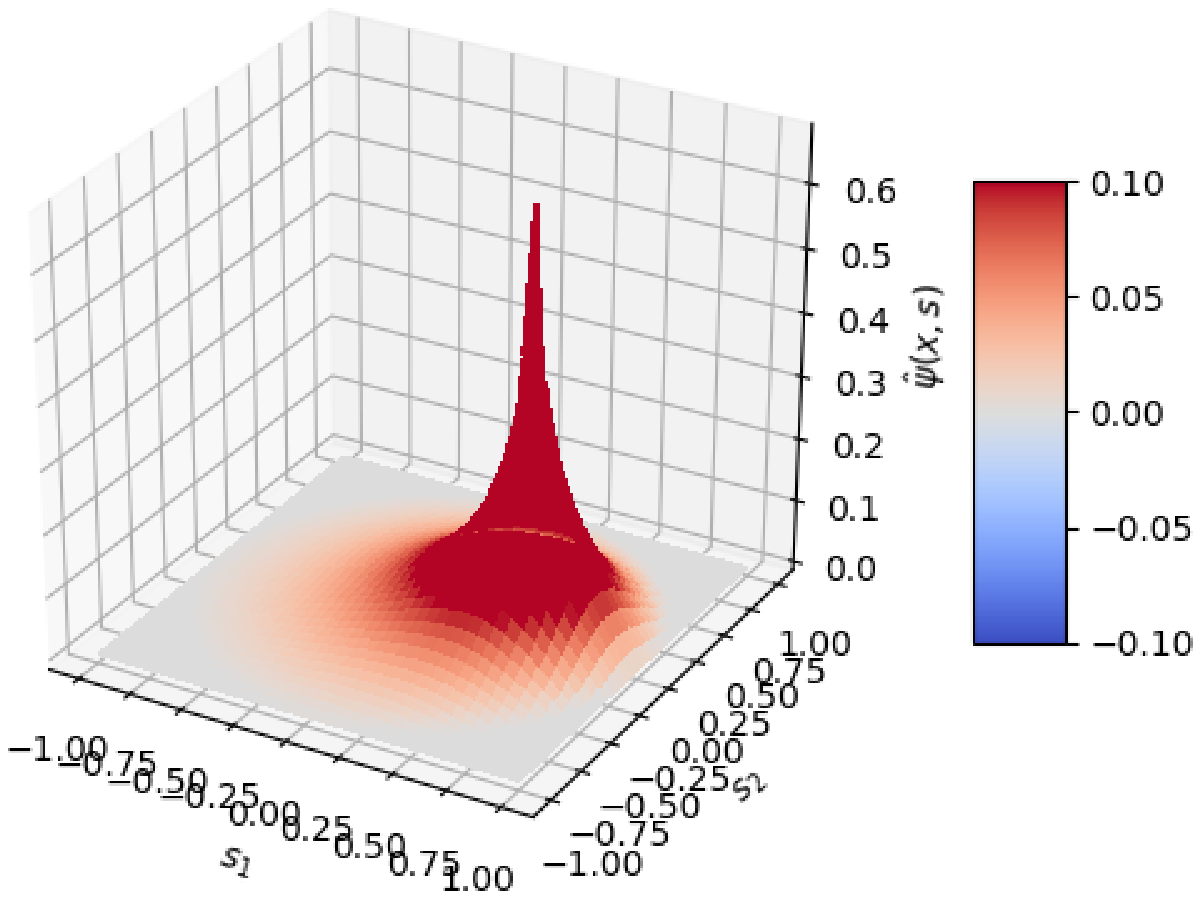}
        \includegraphics[trim=60 20 110 60,clip,width=0.17\textwidth]{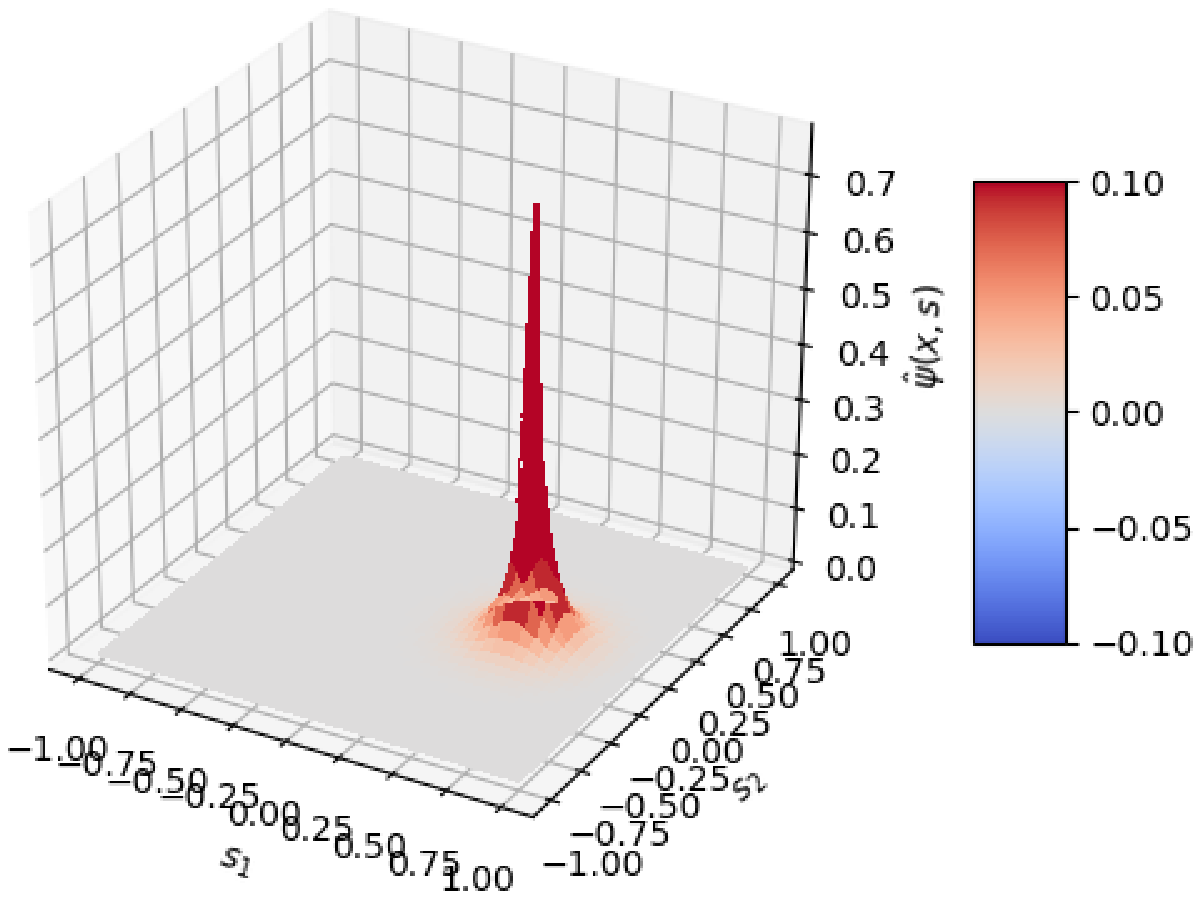}
        % \vspace{-1em}
        \caption{The effect of the parameters $k$, $\lambda$ on the 
        profile of the interaction function $\hat\psi((0,0.5),s)$, $s\in B_2(0,1)$. Left: $(k,\lambda)=(1,0.5)$. Right: $(k,\lambda)=(2,10)$.}
        \label{fig:psi_param_high}
\end{figure}

%%%%%%%%%%%%%%%%%%%%%%%%%%%%%%%%%%%%%%%%%%%%%%%%%%%%%%%%%%%%%%%%%%%%%%%%%%%%%%%%
\section{Conclusion}
\label{Sec:Conclusion}

A family of compactly supported 
parametric interaction functions in the general 
Cucker-Smale flocking dynamics was proposed such that the  
macroscopic system of mass and momentum balance 
equations with non-local damping terms
can be converted to an augmented system of coupled PDEs 
in a compact set.
We approached the computation of the non-local damping
using the standard finite difference treatment of the chosen differential operator, which was solved using banded matrix algorithms.
The expressiveness of the proposed interaction functions 
may be utilized for parametric learning from trajectory data.

%%%%%%%%%%%%%%%%%%%%%%%%%%%%%%%%%%%%%%%%%%%%%%%%%%%%%%%%%%%%%%%%%%%%%%%%%%%%%%%%
\bibliography{bib_pde_flocking,intro_refs}             % bib file to produce the bibliography
                                                     % with bibtex (preferred)
                                                   
%\begin{thebibliography}{xx}  % you can also add the bibliography by hand

%\bibitem[Able(1956)]{Abl:56}
%B.C. Able.
%\newblock Nucleic acid content of microscope.
%\newblock \emph{Nature}, 135:\penalty0 7--9, 1956.

%\end{thebibliography}

%%%%%%%%%%%%%%%%%%%%%%%%%%%%%%%%%%%%%%%%%%%%%%%%%%%%%%%%%%%%%%%%%%%%%%%%%%%%%%%%
%\appendix
%\section{A summary of Latin grammar}    % Each appendix must have a short title.
%\section{Some Latin vocabulary}              % Sections and subsections are supported  
                                                                         % in the appendices.
\end{document}